\documentclass[12pt]{article}
\usepackage[toc,page]{appendix}
\usepackage[usenames]{color}
\usepackage{amssymb}
\usepackage{graphicx}
\usepackage{subfigure}
\usepackage{xcolor}
\usepackage{changebar}
\usepackage{epstopdf}

\newtheorem{proposition}{Proposition}

\newtheorem{example}{Example}

\newenvironment{proof}{\noindent{\bf Proof:}}{\hfill\fbox{}\vspace*{1mm}}
\RequirePackage[normalem]{ulem}

\providecommand{\DIFdeltex}[1]{{\protect\color{red}\sout{#1}}}

\newif\ifdiff
\difftrue
\ifdiff

\newcommand{\del}[1]{\DIFdeltex{#1}}

\else

\newcommand{\del}[1]{}

\fi

\begin{document}

\title{\textbf{Optimal Portfolio Liquidation and Dynamic Mean-variance
Criterion}}
\author{ Jia-Wen Gu\thanks{
Department of Mathematical Science, University of Copenhagen, Copenhagen,
Denmark. (email: kaman.jwgu@gmail.com).} \and Mogens Steffensen \thanks{%
Department of Mathematical Science, University of Copenhagen, Copenhagen,
Denmark. (email: mogens@math.ku.dk).}}
\maketitle

\begin{abstract}
In this paper, we consider the optimal portfolio liquidation problem under
the dynamic mean-variance criterion and derive time-consistent solutions in
three important models. We give adapted optimal strategies under a reconsidered
mean-variance subject at any point in time. We get explicit trading strategies in
the basic model and when random pricing signals are incorporated. When we
consider stochastic liquidity and volatility, we construct a generalized HJB
equation under general assumptions for the parameters. We obtain an explicit
solution in stochastic volatility model with a given structure supported by
empirical studies.
\end{abstract}

\section{Introduction}

As quantitative trading is generally used by financial institutions and
hedge funds, the transactions are usually large in size and may involve the
purchase and sale of hundreds of thousands of shares and other securities.
However, quantitative trading is also commonly used by individual investors.
A fundamental part of agency algorithmic trading in equities and other asset
classes is trade scheduling. Given a trade target, that is, a number of
shares that must be bought or sold before a fixed time horizon, trade
scheduling means planning how many shares will be bought or sold by each
time instant between the beginning of trading and the horizon. This is done
so as to optimize some measure of execution quality, usually measured as the
final average execution price relative to some benchmark price. Almgren and
Chriss (2000) consider the execution of portfolio transactions with the aim
of minimizing a combination of volatility risk and transaction costs arising
from permanent and temporary market impact. Kharroubi and Huy\^{e}n
Pham(2010) study the optimal portfolio liquidation problem over a finite
horizon in a limit order book with bid-ask spread and temporary market price
impact penalizing speedy execution trades, respectively. Almgren (2012)
considers the problem of mean-variance optimal agency execution strategies,
when the market liquidity and volatility vary randomly in time. He
constructs an HJB equation relying on "small impact approximation" under
specific assumptions for the stochastic process satisfied by these
parameters.

The mean-variance analysis of Markowitz (1952) has long been recognized as
the cornerstone of modern portfolio theory. Attention is regained by
relating dynamic mean-variance optimization formalistically to quadratic
utility in Korn and Trautmann (1995), Korn (1997) and Zhou and Li (2000).
The same problem is categorized as mean-variance optimization with
pre-commitment (See Christiansen and Steffensen (2013) for detailed
illutrations). Recently, attention has been regained by Basak and Chabakauri
(2010) who challenge the pre-commitment (to the time 0-expected value as the
target of the quadratic utility) assumed by Zhou and Li (2000). They solve
the problem for the so-called sophisticated investor who updates his
non-linear mean-variance objective and takes future updates,
time-consistently, into account. In this paper, the dynamic mean-variance
criterion is applied to the optimal trading problem.

In this paper, we consider the quantitative trading problem under the
dynamic mean-variance criterion and derive time-consistent solutions in
three important models. Our paper contributes to the quantitative trading
literature in various aspects. Firstly, we solve the dynamic mean-variance
quantitative trading problems and derive time-consistent solutions. We give
optimal strategies under a reconsidered mean-variance subject at any point in time.
 Previous
literature seems only give precommitment and deterministic control
solutions. Almgren (2012) gives the trading strategy in the basic model
where it is fixed rather than adaptive.

Secondly, we determine the explicit solutions when random pricing signals
are incorporated. A random pricing signal, gathering the information of the
index data, trading volume and public and private market events, can be
regarded as the indicator of the stock movements. Various methods have been
proposed to study the pricing signal in the literature. Introduction to the
literature is deferred to Section 3. In this paper, the trading strategy is
derived when the random pricing signal is assumed to be a diffusion process.

Thirdly, we consider the trading strategy in the case of stochastic
volatility and liquidity impact. We allow the liquidity and volatility to
vary randomly in time and the determined trading strategy is adapted to the
market state. We give the generalized HJB equations in the stochastic
volatility and liquidity impact models while early study reply on a "small
impact approximation" (Almgren (2012)). We also get an explicit solution in
a stochastic volatility model with a given structure supported by empirical
study.

The remainder of the paper is structured as follows. In Section 2, we give
the trading strategy in the basic model {which is adopted from Section 1 of Almgen (2012)}. In Section 3, optimal strategy is
presented when random pricing signals are incorporated. In Section 4, we
consider stochastic volatility and liquidity model. Conclusions are given in
Section 5.

\section{The Basic Model}

In this section, we consider the basic model adopted from Almgren (2012). In
the model, the price of a stock is govern by the SDE,%
\begin{equation}
dS(t)=\sigma (t)dW(t),  \label{eq4s}
\end{equation}%
where $\sigma (t)$ is the time-dependent volatility of the stock and $W(t)$
is a standard Brownian motion. The price actually received on each trade is
\[
\tilde{S}(t)=S(t)+\eta (t)\upsilon (t),
\]%
where $\eta (t)$ is the coefficient of temporary market impact, also
time-varying and $\upsilon $ is the rate of buying. The volatility and
impact functions $\sigma (t)$ and $\eta (t)$ are assumed to be continuous
functions of $t$ to account for trading seasonality. The trader begins at
time $t=0$ with a purchase target of $x$ shares, which must be completed by
time $t=T$. The number of shares yet to be purchased at time $t$ is the
trajectory $X(t)$, with $X(0)=x$ and $X(T)=0$. Hence
\begin{equation}
\upsilon (t)=-\frac{dX(t)}{dt}.  \label{eq4x}
\end{equation}%
The cost of trading, is the total dollars paid to purchase $x$ shares
subtracting the initial market value:
\[
C=\int_{0}^{T}\tilde{S}(t)\upsilon (t)dt-xS(0).
\]%
By integration by parts, we rewrite
\begin{equation}
C=\int_{0}^{T}\eta (t)\upsilon ^{2}(t)dt+\int_{0}^{T}\sigma (t)X(t)dW(t).
\end{equation}%
We determine the optimal trajectory by the dynamic mean-variance criterion
\[
\min_{\upsilon (s):t\leq s\leq T}E_{t}(C)+\mu Var_{t}(C).
\]%
This is newly proposed by Basak and Chabakuri (2010), who come up with the
dynamic mean-variance criterion challenging the pre-commitment mean-variance
assumed by Korn (1997) and Zhou and Li (2000). Basak and Chabakuri (2010)
use this criterion in asset allocation problem for the so-called
sophisticated investor who updates his nonlinear mean-variance objective and
takes future updates, time-consistently, into account.

To address the optimal trading problem, we define a process $Y$ by
\begin{equation}
\begin{array}{lll}
Y(0) & = & 0, \\
dY(t) & = & \eta (t)\upsilon ^{2}(t)dt+\sigma (t)X(t)dW(t),%
\end{array}
\label{eq4y}
\end{equation}%
such that $C=Y\left( T\right) $. Our objective becomes
\[
\min_{\upsilon (s):t\leq s\leq T}E_{t}(Y(T))+\mu Var_{t}(Y(T)).
\]%
By (\ref{eq4y}), we know
\begin{equation}
\begin{array}{ll}
& E_{t}(Y(T))+\mu Var_{t}(Y(T)) \\
= & Y(t)+E_{t}(\int_{t}^{T}\sigma (s)X(s)dW(s)+\int_{t}^{T}\eta (s)\upsilon
^{2}(s)ds) \\
& +\mu Var_{t}(\int_{t}^{T}\sigma (s)X(s)dW(s)+\int_{t}^{T}\eta (s)\upsilon
^{2}(s)ds)%
\end{array}
\label{cal}
\end{equation}%
Suppose that we are given an optimal trading strategy $\upsilon ^{\ast
}(s),t\leq s\leq T$, and the corresponding value of $Y^{\ast }$. The value
function is defined as
\begin{equation}
J(t)=J(S(t),X(t),Y(t),t)=E_{t}(Y^{\ast }(T))+\mu Var_{t}(Y^{\ast }(T)).
\label{eq4j'}
\end{equation}%
Noting that by the law of total variance,%
\begin{equation}
Var_{t}(Y^{\ast }(T))=E_{t}(Var_{t+\tau }(Y^{\ast }(T)))+Var_{t}(E_{t+\tau
}(Y^{\ast }(T))).  \label{totalvariance'}
\end{equation}%
Plugging (\ref{totalvariance'}) into (\ref{eq4j'}), we obtain
\begin{equation}
J(t)=\min_{\upsilon (s):t\leq s\leq t+\tau }E_{t}(J(t+\tau ))+\mu
Var_{t}(E_{t+\tau }(Y^{\ast }(T))).  \label{eq4j2'}
\end{equation}%
From (\ref{eq4y}), we know
\begin{equation}
E_{t}(Y^{\ast }(T))=Y(t)+f(t),  \label{eq4y*'}
\end{equation}%
where
\begin{equation}
f(t)=E_{t}\left( \int_{t}^{T}\eta (\upsilon ^{\ast }(s))^{2}ds\right) .
\label{eq4f'}
\end{equation}%
Besides, plugging (\ref{eq4y}) into (\ref{eq4j'}) we also know
\begin{equation}
J(t)=Y(t)+C(t),  \label{eq4separation'}
\end{equation}%
where
\begin{equation}
\begin{array}{lll}
C(t) & = & C(S(t),X(t),t) \\
& = & E_{t}\left( \int_{t}^{T}\eta (s)(\upsilon ^{\ast }(s))^{2}ds+\sigma
(s)X^{\ast }(s)dW(s)\right)  \\
&  & +\mu Var_{t}\left( \int_{t}^{T}\eta (s)(\upsilon ^{\ast
}(s))^{2}ds+\sigma (s)X^{\ast }(s)dW(s)\right)
\end{array}
\label{negligible}
\end{equation}%
does not depend on $Y(t)$.

\begin{proposition}
The HJB equation concerning the optimal trading problem is given by
\begin{equation}  \label{hjb'}
0 = \frac{1}{2} C_{s s}\sigma^2 + C_t + \min_{\upsilon} \{\eta
\upsilon^2-C_x \upsilon\} +\mu \sigma^2(x+f_{s})^2
\end{equation}
where the minimum is clearly $\upsilon^* = \frac{C_x}{2 \eta}$ and $f$
satisfies
\begin{equation}  \label{eq4f2'}
0= \eta (\upsilon^*)^2 + \frac{1}{2} f_{s s}\sigma^2 -f_x \upsilon^* +f_t.
\end{equation}
\end{proposition}

\begin{proof}
Combining (\ref{eq4j2'}), (\ref{eq4y*'}) and (\ref{eq4separation'}), we
obtain
\[
0=\min_{\upsilon }E_{t}(dY(t)+dC(t))+\mu Var_{t}(dY(t)+df(t)).
\]%
Using the It$\hat{\mathrm{o}}$'s lemma, (\ref{eq4x}), (\ref{eq4y}) and after inserting $X(t)=x, S(t)=s$,
\[
\begin{array}{lll}
0 & = & \min_{\upsilon }\{\eta \upsilon ^{2}dt+\frac{1}{2}C_{ss}\sigma
^{2}dt-C_{x}\upsilon dt+C_{t}dt+\mu Var_{t}(\sigma xdW+\sigma f_{s}dW)\} \\
& = & \min_{\upsilon }\{\eta \upsilon ^{2}dt+\frac{1}{2}C_{ss}\sigma
^{2}dt-C_{x}\upsilon dt+C_{t}dt+\mu \sigma ^{2}(x+f_{s})^{2}dt\},%
\end{array}%
\]%
hence (\ref{hjb'}) follows and the optimal strategy is given by $\upsilon
^{\ast }=\frac{C_{x}}{2\eta }$. Together with (\ref{eq4f'}), (\ref{eq4f2'})
follows.
\end{proof}

Consequently, the two PDEs for $C$ and $f$ can be derived as follows
\begin{eqnarray}
0 &=&\frac{1}{2}C_{ss}\sigma ^{2}+C_{t}-\frac{C_{x}^{2}}{4\eta }+\mu \sigma
^{2}(x+f_{s})^{2},  \label{eq4pdec'} \\
0 &=&\frac{C_{x}^{2}}{4\eta }+\frac{1}{2}f_{ss}\sigma ^{2}-\frac{C_{x}f_{x}}{%
2\eta }+f_{t}.  \label{eq4pdef'}
\end{eqnarray}

From (\ref{cal}), we see that the optimal strategy $\upsilon ^{\ast }(s)$
does not depend on $Y(t)$ and $S(t)$ for $s\geq t$. Hence, $\upsilon ^{\ast
}(s)$ depends only on $X(s)$ and $s$. Combining (\ref{eq4x}), $\upsilon
^{\ast }$ is a deterministic control and $X$ is a deterministic process.
Therefore, (\ref{eq4pdec'}) reduced to
\begin{equation}
0=\mu \sigma ^{2}x^{2}+C_{t}-\frac{C_{x}^{2}}{4\eta }.  \label{eq4pde2}
\end{equation}%
The initial data for the PDE (\ref{eq4pde2}) is a local asymptotic
condition. Considering (\ref{negligible}), near expiration $T$, the terms
with $dW$ become negligible, then we must liquidate on a linear trajectory $%
\upsilon =x/(T-t)$ and hence the function $C$ has local behavior
\[
C\sim \frac{\eta (t)x^{2}}{T-t}+O(T-t),\ T-t\rightarrow 0.
\]%
We look for a candidate solution to HJB in the form $C=x^{2}L(t).$ Plugging
into the HJB, we see that $L$ should satisfy the ODE:
\[
0=\mu \sigma (t)^{2}+L^{\prime }(t)-\frac{L^{2}(t)}{\eta (t)}
\]%
and
\[
L\sim \frac{\eta (t)}{T-t}+O(T-t),\ T-t\rightarrow 0.
\]%
The optimal strategy is given by
\[
\upsilon ^{\ast }(t, X(t))=\frac{1}{\eta(t)}X(t)L(t).
\]
If we restrict $\eta $ and $\theta $ to be constant,
\[
L(t)=\sqrt{\mu \eta \sigma ^{2}}\coth \left( \sqrt{\frac{\mu \sigma ^{2}}{%
\eta }}(T-t)\right) .
\]%
The optimal strategy can be expressed as
\[
\upsilon ^{\ast }(t, X(t))=X(t)\sqrt{\frac{\mu \sigma ^{2}}{\eta }}\coth \left(
\sqrt{\frac{\mu \sigma ^{2}}{\eta }}(T-t)\right) .
\]
{
This is the same strategy as the one obtained by Almgren (2012). However, it is important to realize that the problem formulations are different. Whereas Almgren (2012) finds the best deterministic strategy for a classical mean-variance problem, we find the best stochastic strategy for a time-consistent formulation of the mean-variance problem. Since the best strategy is deterministic, time-consistency does not distinguish the two problems and the two strategies coincide. However, when we proceed and add randomness from signals, volatility, and liquidity, this coincidence is lost, since our strategies become adapted and reflect specifically the time-consistency of the problem formulation.}

\section{Random Pricing Signals}

{We consider the trading problem of mean-variance optimal agency execution
strategies, when a random pricing signal is included. A random pricing
signal, gathering the information of index data, trading volume and public
and private market events, can be regarded as the indicator of the stock
movement. Various research work have considered pricing signals for the
support and prediction of limit and market order placement strategies of
traders. Interested readers are advised to refer to Milgrom and Stokey
(1982) and Suominen (2001). }

The model with random pricing signals enhances the trading quantity for two
reasons. First, the incorporation of pricing signals relates stock returns
to market returns. One can identify the pricing signals by investigating
statistical and normal relationships between an asset's returns and market
factors. Some notable examples of understanding the relationships between
stock returns and market returns include the CAMP model by Sharpe (1964),
the common risk factor model by Fama and French (1993), the Extended
four-factor model by Carhart (1997) and the GARCH model by Lamoureux and
Lastrapes (1990). Second, the model replies on the belief that extreme price
movements are caused by temporary liquidity shortage and manipulation and
would be followed by a price reversal, which is consistent with the market
behavior. In our model, the price reversal is described be a reverting
process with rate $\theta (t)$.

Another example of the incorporation of random pricing signals is pairs
trading. The strategy monitors performance of two historically correlated
securities, e.g. Coca-Cola (KO) and Pepsi (PEP). When the correlation
between the two securities temporarily weakens, i.e. one stock moves up
while the other moves down, the pairs trade would be to short the
outperforming stock and to long the underperforming one, betting that the
"spread" between the two would eventually converge (See Mudchanatongsuk et
al. (2008)). One can identify the pricing signals by investigating the
average stock movements of the pair of stocks and finding the optimal
trading strategy under the mean-variance criterion by our approach.

To be more specific, the price of a stock is govern by SDEs:
\begin{equation}
\begin{array}{lll}
dS(t) & = & \theta (t)(\alpha (t)-S(t))dt+\sigma _{1}(t)dW_{1}(t), \\
d\alpha (t) & = & \sigma _{2}(t)dW_{2}(t),%
\end{array}
\label{eq4signal}
\end{equation}%
where $\alpha (t)$ is a random pricing signal, $\theta (t)$ the rate by
which the shock dissipate and the variable reverts towards the signal, $%
\sigma _{1}(t)$ is the volatility of the stock and $W_{1}(t)$ and $W_{2}(t)$
are independent standard Brownian motions. We assume $\theta (t)$, $\eta (t)$%
, $\sigma _{1}(t)$ and $\sigma _{2}(t)$ are continuously time-varying to
account for trading seasonality.

The cost of trading, is the total dollars paid to purchase $X$ shares
subtracting the initial market value:
\[
\begin{array}{lll}
C & = & \int_{0}^{T}\tilde{S}(t)\upsilon (t)dt-xS(0) \\
& = & \int_{0}^{T}\theta (t)X(t)(\alpha (t)-S(t))dt+\int_{0}^{T}\eta
(t)\upsilon ^{2}(t)dt+\int_{0}^{T}\sigma _{1}(t)X(t)dW_{1}(t).%
\end{array}%
\]%
We determine the optimal trajectory by the dynamic mean-variance criterion
\[
\min_{\upsilon (s):t\leq s\leq T}E_{t}(C)+\mu Var_{t}(C).
\]%
We now follow the recipe presented in the previous section and define a
process $Y$ as
\begin{equation}
\begin{array}{lll}
Y(0) & = & 0, \\
dY(t) & = & \eta (t)\upsilon ^{2}(t)dt+\theta (t)X(t)(\alpha
(t)-S(t))dt+\sigma _{1}(t)X(t)dW_{1}(t),%
\end{array}
\label{eq4y2}
\end{equation}%
such that $Y\left( T\right) =C$. Our objective becomes
\[
\min_{\upsilon (s):t\leq s\leq T}E_{t}(Y(T))+\mu Var_{t}(Y(T)).
\]%
Suppose that we are given an optimal trading strategy $\upsilon ^{\ast
}(s),t\leq s\leq T$, and the corresponding value of $Y^{\ast }$. The value
function is defined as
\begin{equation}
J(t)=J(S(t),\alpha (t),X(t),Y(t),t)=E_{t}(Y^{\ast }(T))+\mu Var_{t}(Y^{\ast
}(T)).  \label{eq4j}
\end{equation}%
By the law of total variance, we obtain
\begin{equation}
J(t)=\min_{\upsilon (s):t\leq s\leq t+\tau }E_{t}(J(t+\tau ))+\mu
Var_{t}(E_{t+\tau }(Y^{\ast }(T))).  \label{eq4j2}
\end{equation}%
From (\ref{eq4y2}), we know
\begin{equation}
E_{t}(Y^{\ast }(T))=Y(t)+f(t),  \label{eq4y*}
\end{equation}%
where
\begin{equation}
f(t)=E_{t}\left( \int_{t}^{T}\theta (s)(\alpha (s)-S(s))X^{\ast }(s)ds+\eta
(s)(\upsilon ^{\ast }(s))^{2}ds\right) .  \label{eq4f}
\end{equation}%
Besides, plugging (\ref{eq4y2}) into (\ref{eq4j}) we also know
\begin{equation}
J(t)=Y(t)+C(t),  \label{eq4separation}
\end{equation}%
where
\[
\begin{array}{lll}
C(t) & = & C(S(t),\alpha (t),X(t),t) \\
& = & E_{t}\left( \int_{t}^{T}\theta (s)(\alpha (s)-S(s))X^{\ast }(s)ds+\eta
(s)(\upsilon ^{\ast }(s))^{2}ds+\sigma _{1}(s)X^{\ast }(s)dW_{1}(s)\right)
\\
&  & +\mu Var_{t}\left( \int_{t}^{T}\theta (s)(\alpha (s)-S(s))X^{\ast
}(s)ds+\eta (s)(\upsilon ^{\ast }(s))^{2}ds+\sigma _{1}(s)X^{\ast
}(s)dW_{1}(s)\right),
\end{array}%
\]%
does not depend on $Y(t)$. To proceed, we reduce the dimension of the HJB by
defining a new variable:
\[
\beta (t)=S(t)-\alpha (t),
\]%
the difference between the stock price and the observable signal. We can
easily see that function $f(t)$ and $C(t)$ depend only on $\beta (t),X(t),t$
and
\begin{eqnarray}
dY(t) &=&-\theta (t)\beta (t)X(t)dt+\eta (t)\upsilon ^{2}(t)dt+\sigma
_{1}(t)X(t)dW_{1}(t),  \label{eq4y3} \\
d\beta (t) &=&-\theta (t)\beta (t)dt+\sigma _{1}(t)dW_{1}(t)-\sigma
_{2}(t)dW_{2}(t).  \label{eq4beta}
\end{eqnarray}

\begin{proposition}
The HJB equation concerning the optimal trading problem with random signal
is given by
\begin{equation}  \label{hjb}
0 = -\theta x \beta - C_{\beta} \theta \beta +\frac{1}{2} C_{\beta
\beta}(\sigma^2_1 +\sigma^2_2 ) + C_t + \min_{\upsilon} \{\eta
\upsilon^2-C_x \upsilon\} +\mu \sigma^2_1(x+f_{\beta})^2 +\mu \sigma^2_2
f^2_{\beta},
\end{equation}
where the minimum is clearly $\upsilon^* = \frac{C_x}{2 \eta}$ and $f$
satisfies
\begin{equation}  \label{eq4f2}
0= -\theta x \beta + \eta (\upsilon^*)^2 -f_{\beta} \theta \beta + \frac{1}{2%
} f_{\beta \beta}(\sigma^2_1 +\sigma^2_2 ) -f_x \upsilon^* +f_t.
\end{equation}
\end{proposition}

\begin{proof}
Combining (\ref{eq4j2}), (\ref{eq4y*}) and (\ref{eq4separation}), we obtain
\[
0 = \min_{\upsilon} E_t(dY(t) +dC(t)) + \mu Var_t (dY(t)+df(t)).
\]
Using the It$\hat{\mathrm{o}}$'s lemma, (\ref{eq4x}), (\ref{eq4y3}), (\ref%
{eq4beta}) and after inserting $X(t)=x, \beta(t)=\beta$,
\[
\begin{array}{lll}
0 & = & \min_{\upsilon} \{-\theta x \beta dt +\eta \upsilon^2 dt -C_{\beta}
\theta \beta dt + \frac{1}{2} C_{\beta \beta}(\sigma^2_1 +\sigma^2_2 )dt
-C_x \upsilon dt +C_t dt + \\
&  & \mu Var_t (\sigma_1 x dW_1 + \sigma_1 f_{\beta}dW_1 - \sigma_2
f_{\beta}dW_2)\} \\
& = & \min_{\upsilon} \{-\theta x \beta dt +\eta \upsilon^2 dt -C_{\beta}
\theta \beta dt + \frac{1}{2} C_{\beta \beta}(\sigma^2_1 +\sigma^2_2 )dt
-C_x \upsilon dt +C_t dt + \\
&  & \mu \sigma^2_1 (x+f_{\beta})^2dt +\mu \sigma^2_2 f^2_{\beta}dt \},%
\end{array}
\]
hence (\ref{hjb}) follows and the optimal strategy is given by $\upsilon^* =
\frac{C_x}{2 \eta}$. Together with (\ref{eq4f}), (\ref{eq4f2}) follows.
\end{proof}

Consequently, the two PDEs for $C$ and $f$ can be derived as follows
\begin{eqnarray}
0 &=&-\theta x\beta -C_{\beta }\theta \beta +\frac{1}{2}C_{\beta \beta
}(\sigma _{1}^{2}+\sigma _{2}^{2})+C_{t}-\frac{C_{x}^{2}}{4\eta }+\mu \sigma
_{1}^{2}(x+f_{\beta })^{2}+\mu \sigma _{2}^{2}f_{\beta }^{2}.  \label{eq4pdec}
\\
0 &=&-\theta x\beta +\frac{C_{x}^{2}}{4\eta }-f_{\beta }\theta \beta +\frac{1%
}{2}f_{\beta \beta }(\sigma _{1}^{2}+\sigma _{2}^{2})-\frac{C_{x}f_{x}}{%
2\eta }+f_{t}.  \label{eq4pdef}
\end{eqnarray}%
Similar to the local asymptotic condition in the basic model, near
expiration, one must liquidate on a linear trajectory. Therefore,
\[
C\sim \frac{\eta (t)x^{2}}{T-t}+O(T-t),f\sim \frac{\eta (t)x^{2}}{T-t}%
+O(T-t),\ T-t\rightarrow 0.
\]%
We look for a candidate solution to PDEs in the form
\begin{equation}
\begin{array}{lll}
C & = & x^{2}D(t)+\beta ^{2}E(t)+x\beta F(t)+xG(t)+\beta H(t)+I(t), \\
f & = & x^{2}L(t)+\beta ^{2}M(t)+x\beta N(t)+xO(t)+\beta P(t)+Q(t).%
\end{array}
\label{eq4cf}
\end{equation}%
We plug (\ref{eq4cf}) into (\ref{eq4pdec}) (\ref{eq4pdef}) and obtain a
system of ODEs.
\begin{equation}
\begin{array}{llll}
-\frac{dD}{dt} & = & -\frac{1}{\eta }D^{2}+\mu \sigma _{1}^{2}(1+N)^{2}+\mu
\sigma _{2}^{2}N^{2}, & D\sim \frac{\eta }{T-t},t\rightarrow T, \\
-\frac{dE}{dt} & = & -2\theta E-\frac{1}{4\eta }F^{2}+4\mu (\sigma
_{1}^{2}+\sigma _{2}^{2})M^{2}, & E(T)=0, \\
-\frac{dF}{dt} & = & -\theta -\theta F-\frac{1}{\eta }DF+4\mu \sigma
_{1}^{2}(1+N)M+4\mu \sigma _{2}^{2}MN, & F(T)=0, \\
-\frac{dG}{dt} & = & -\frac{1}{\eta }DG+2\mu (\sigma _{1}^{2}+\sigma
_{2}^{2})NP, & G(T)=0, \\
-\frac{dH}{dt} & = & -\theta H-\frac{1}{2\eta }FG+4\mu (\sigma
_{1}^{2}+\sigma _{2}^{2})MP, & H(T)=0, \\
-\frac{dI}{dt} & = & (\sigma _{1}^{2}+\sigma _{2}^{2})E-\frac{1}{4\eta }%
G^{2}+\mu (\sigma _{1}^{2}+\sigma _{2}^{2})P^{2}, & I(T)=0, \\
-\frac{dL}{dt} & = & \frac{1}{\eta }D^{2}-\frac{2}{\eta }DL, & L\sim \frac{%
\eta }{T-t},t\rightarrow T, \\
-\frac{dM}{dt} & = & \frac{1}{4\eta }F^{2}-2\theta M-\frac{1}{2\eta }FN, &
M(T)=0, \\
-\frac{dN}{dt} & = & -\theta +\frac{1}{\eta }DF-\theta N-\frac{1}{\eta }%
(DN+FL), & N(T)=0, \\
-\frac{dO}{dt} & = & \frac{1}{\eta }(DG-DO-GL), & O(T)=0, \\
-\frac{dP}{dt} & = & -\theta P+\frac{1}{2\eta }(FG-FO-GN), & P(T)=0, \\
-\frac{dQ}{dt} & = & \frac{1}{4\eta }G^{2}-(\sigma _{1}^{2}+\sigma
_{2}^{2})M-\frac{1}{2\eta }GO, & Q(T)=0.%
\end{array}
\label{ODEs}
\end{equation}%
From (\ref{ODEs}), one can find that solutions to $G,O,P$ are trivial, i.e.,
$G=O=P=0$. Hence the optimal strategy becomes
\[
\upsilon ^{\ast }=\frac{1}{2\eta }(2Dx+F\beta ),
\]%
which only evolves $D$ and $F$. Therefore, (\ref{ODEs}) can be reduced to (%
\ref{ODEss}), which gives the optimal strategy.
\begin{equation}
\begin{array}{llll}
-\frac{dD}{dt} & = & -\frac{1}{\eta }D^{2}+\mu \sigma _{1}^{2}(1+N)^{2}+\mu
\sigma _{2}^{2}N^{2}, & D\sim \frac{\eta }{T-t},t\rightarrow T, \\
-\frac{dF}{dt} & = & -\theta -\theta F-\frac{1}{\eta }DF+4\mu \sigma
_{1}^{2}(1+N)M+4\mu \sigma _{2}^{2}MN, & F(T)=0, \\
-\frac{dL}{dt} & = & \frac{1}{\eta }D^{2}-\frac{2}{\eta }DL, & L\sim \frac{%
\eta }{T-t},t\rightarrow T, \\
-\frac{dM}{dt} & = & \frac{1}{4\eta }F^{2}-2\theta M-\frac{1}{2\eta }FN, &
M(T)=0, \\
-\frac{dN}{dt} & = & -\theta +\frac{1}{\eta }DF-\theta N-\frac{1}{\eta }%
(DN+FL), & N(T)=0.%
\end{array}
\label{ODEss}
\end{equation}
To summarize, the optimal strategy is given by
\[
\upsilon^*(t, X(t), \beta(t)) =\frac{1}{2 \eta(t)} (2 D(t)X(t)+F(t)\beta(t)).
\]

\begin{example}
If we set $\theta=0$, then this model reduces to be the basic one. From (\ref%
{ODEs}), one can find solutions to $F, M, N$ are trivial if $\theta=0$,
i.e., $F=M=N=0$. Consequently,
\[
\begin{array}{llll}
-\frac{d D}{d t} & = & -\frac{1}{\eta}D^2 + \mu \sigma^2_1, & D \sim \frac{%
\eta}{T-t}, t \rightarrow T, \\
-\frac{d L}{d t} & = & \frac{1}{\eta} D^2 -\frac{2}{\eta} DL, & L \sim \frac{%
\eta}{T-t}, t \rightarrow T. \\
&  &  &
\end{array}
\]
If we restrict $\sigma_1$ and $\eta$ to be constant, we obtain,
\[
D(t)= \sqrt{\mu \eta \sigma^2_1} \coth \left( \sqrt{\frac{\mu \sigma^2_1}{%
\eta}}(T-t)\right),
\]
and
\[
\upsilon^*(t, X(t))=X(t)\sqrt{\frac{\mu \sigma^2_1}{\eta}} \coth \left( \sqrt{%
\frac{\mu \sigma^2_1}{\eta}}(T-t)\right).
\]
\end{example}

\section{Stochastic Liquidity and Volatility}

In this section, we consider the liquidity impact $\eta (t)$ and $\sigma (t)$
in the basic model to de dependent on the trading position $X(t)$ and an
independent variable $\xi (t)$ representing the \textquotedblleft market
state", i.e.,
\[
d\xi (t)=a_{\xi }(t)dt+b_{\xi }(t)dB\left( t\right),
\]%
where $a_{\xi }$ and $b_{\xi }$ are known function of $t$ and $B$ is a
Brownian motion independent of $W$. A derivation corresponding to the
previous sections leads to the value function
\[
J(t)=J(\xi (t),X(t),t)=E_{t}(Y^{\ast }(T))+\mu Var_{t}(Y^{\ast
}(T))=Y(t)+C(t),
\]%
where
\[
C(t)=C(\xi (t),X(t),t).
\]%
We also have
\[
E_{t}(Y^{\ast }(T))=Y(t)+f(t),
\]%
where
\begin{equation}
f(t)=f(\xi (t),X(t),t)=E_{t}\left[ \int_{t}^{T}\eta (t)(\upsilon ^{\ast
}(t))^{2}dt\right].   \label{eq4f4}
\end{equation}

\begin{proposition}
The HJB equation concerning the optimal trading problem is given by
\begin{equation}  \label{hjb4}
0 = \frac{1}{2} C_{\xi \xi}b^2_{\xi} + C_{\xi}a_{\xi} + C_t +
\min_{\upsilon} \{\eta \upsilon^2-C_x \upsilon\} +\mu \sigma^2x^2 + \mu
b_{\xi} f^2_{\xi},
\end{equation}
where the minimum is clearly $\upsilon^* = \frac{C_x}{2 \eta}$ and $f$
satisfies
\begin{equation}  \label{eq4f3}
0= \eta (\upsilon^*)^2 + \frac{1}{2} f_{\xi \xi}b^2_{\xi} +f_{\xi}
a_{\xi}-f_x \upsilon^* +f_t.
\end{equation}
\end{proposition}

\begin{proof}
As usual we have%
\[
0=\min_{\upsilon }E_{t}(dY(t)+dC(t))+\mu Var_{t}(dY(t)+df(t)).
\]%
Using the It$\hat{\mathrm{o}}$'s lemma and after inserting $X(t)=x, \xi(t)=\xi$,
\[
\begin{array}{lll}
0 & = & \min_{\upsilon }\{\eta \upsilon ^{2}dt+\frac{1}{2}C_{\xi \xi }b_{\xi
}^{2}dt+C_{\xi }a_{\xi }dt-C_{x}\upsilon dt+C_{t}dt+\mu Var_{t}(\sigma
xdW+\sigma f_{\xi }b_{\xi }dB)\} \\
& = & \min_{\upsilon }\{\eta \upsilon ^{2}dt+\frac{1}{2}C_{\xi \xi }b_{\xi
}^{2}dt+C_{\xi }a_{\xi }dt-C_{x}\upsilon dt+C_{t}dt+\mu \sigma
^{2}x^{2}dt+\mu f_{\xi }^{2}b_{\xi }^{2}\},%
\end{array}%
\]%
hence (\ref{hjb4}) follows and the optimal strategy is given by $\upsilon
^{\ast }=\frac{C_{x}}{2\eta }$. Together with (\ref{eq4f4}), (\ref{eq4f3})
follows.
\end{proof}

Consequently, the two PDEs for $C$ and $f$ can be derived as follows
\begin{eqnarray}
0 &=&\frac{1}{2}C_{\xi \xi }b_{\xi }^{2}+C_{\xi }a_{\xi }+C_{t}-\frac{%
C_{x}^{2}}{4\eta }+\mu \sigma ^{2}x^{2}+\mu f_{\xi }^{2}b_{\xi }^{2},
\label{eq4pdec4} \\
0 &=&\frac{C_{x}^{2}}{4\eta }+\frac{1}{2}f_{\xi \xi }b_{\xi }^{2}+f_{\xi
}a_{\xi }-\frac{C_{x}f_{x}}{2\eta }+f_{t}.  \label{eq4pdef4}
\end{eqnarray}

Finding an explicit solution to the system of PDEs is difficult but we can
still find one under some assumptions.

\begin{example}
Here we provide an example to capture stochastic volatility and time-varying
liquidity impact, where we assume $\sigma(t)=\sqrt{\frac{\xi(t)}{ X(t)/(T-t)}%
}$. Blais and Protter (2010) examine the structure of the supply curve using
tick data. They find that for highly liquid stocks, the supply curve is
effectively linear, with a slope that varies with time. Their empirical
analysis also indicates the slope has a small variance. This supports we use
a time-varying liquidity impact $\eta(t)$. Empirical investigations (Jones,
Kaul and Lipson (1994)) reveal a significantly positive relation between
trade size, volume of transactions and stock volatility. If from time $t$ to
$T$, liquidity of the stock is mainly provided by the trader. There is a
positive relation between the volatility and speed of liquidity. The average
speed of liquidity is fixed from $0$ to $T$, i.e., $x/T$. So there is a
negative relation between the volatility at time $t$ and the average speed
of liquidity from $t$ to $T$, i.e., $X(t)/(T-t)$. This supports our
assumption $\sigma(t)=\sqrt{\frac{\xi(t)}{ X(t)/(T-t)}}$.

We look for a solution of the form:
\begin{equation}
\begin{array}{lll}
C & = & x^{2}D(t)+\xi ^{2}E(t)+x\xi F(t)+xG(t)+\xi H(t)+I(t), \\
f & = & x^{2}L(t)+\xi ^{2}M(t)+x\xi N(t)+xO(t)+\xi P(t)+Q(t),%
\end{array}
\label{eq4cf2}
\end{equation}%
with
\[
C\sim \frac{\eta (t)x^{2}}{T-t}+O(T-t),f\sim \frac{\eta (t)x^{2}}{T-t}%
+O(T-t),\ T-t\rightarrow 0.
\]%
Consequently, we obtain a system of ODEs.
\begin{equation}
\begin{array}{llll}
-\frac{dD}{dt} & = & -\frac{1}{\eta }D^{2}+\mu b_{\xi }^{2}N^{2}, & D\sim
\frac{\eta }{T-t},t\rightarrow T, \\
-\frac{dE}{dt} & = & -\frac{1}{4\eta }F^{2}+4\mu b_{\xi }^{2}M^{2}, & E(T)=0,
\\
-\frac{dF}{dt} & = & -\frac{1}{\eta }DF+4\mu b_{\xi }^{2}MN+\mu (T-t), &
F(T)=0, \\
-\frac{dG}{dt} & = & -\frac{1}{\eta }DG+2\mu b_{\xi }^{2}NP+a_{\xi }F, &
G(T)=0, \\
-\frac{dH}{dt} & = & -\frac{1}{2\eta }FG+4\mu b_{\xi }^{2}MP+2a_{\xi }E, &
H(T)=0, \\
-\frac{dI}{dt} & = & b_{\xi }^{2}E-\frac{1}{4\eta }G^{2}+\mu b_{\xi
}^{2}P^{2}+a_{\xi }H, & I(T)=0, \\
-\frac{dL}{dt} & = & \frac{1}{\eta }D^{2}-\frac{2}{\eta }DL, & L\sim \frac{%
\eta }{T-t},t\rightarrow T, \\
-\frac{dM}{dt} & = & \frac{1}{4\eta }F^{2}-\frac{1}{2\eta }FN, & M(T)=0, \\
-\frac{dN}{dt} & = & \frac{1}{\eta }DF-\frac{1}{\eta }(DN+FL), & N(T)=0, \\
-\frac{dO}{dt} & = & \frac{1}{\eta }(DG-DO-GL)+a_{\xi }N, & O(T)=0, \\
-\frac{dP}{dt} & = & \frac{1}{2\eta }(FG-FO-GN)+2a_{\xi }M, & P(T)=0, \\
-\frac{dQ}{dt} & = & \frac{1}{4\eta }G^{2}-\frac{1}{2\eta }GO+a_{\xi
}P+b_{\xi }^{2}M, & Q(T)=0.%
\end{array}%
\end{equation}

The optimal strategy becomes
\[
\upsilon^*(t, X(t), \xi(t))=\frac{1}{2\eta(t)}(2D(t)X(t)+F(t)\xi(t)+G(t)).
\]
\end{example}

\section{Numerical Illustrations}

In this section, we give numerical examples of our three models for
quantitative trading with dynamic mean variance criterion. Our trading
target is to buying 100 shares of stock ($x=100$) within a week (5 working
days, i.e. $T=5$) and we set parameter $\mu=1$.

Figure 1 gives the trade strategy in the basic model with various constant
values for volatility and liquidity impact.

\begin{figure}[]
\caption{Trade Strategy in the Basic Model with various values for
volatility and liquidity Impact}
\begin{center}
\subfigure[]{
\resizebox*{13cm}{!}{\includegraphics{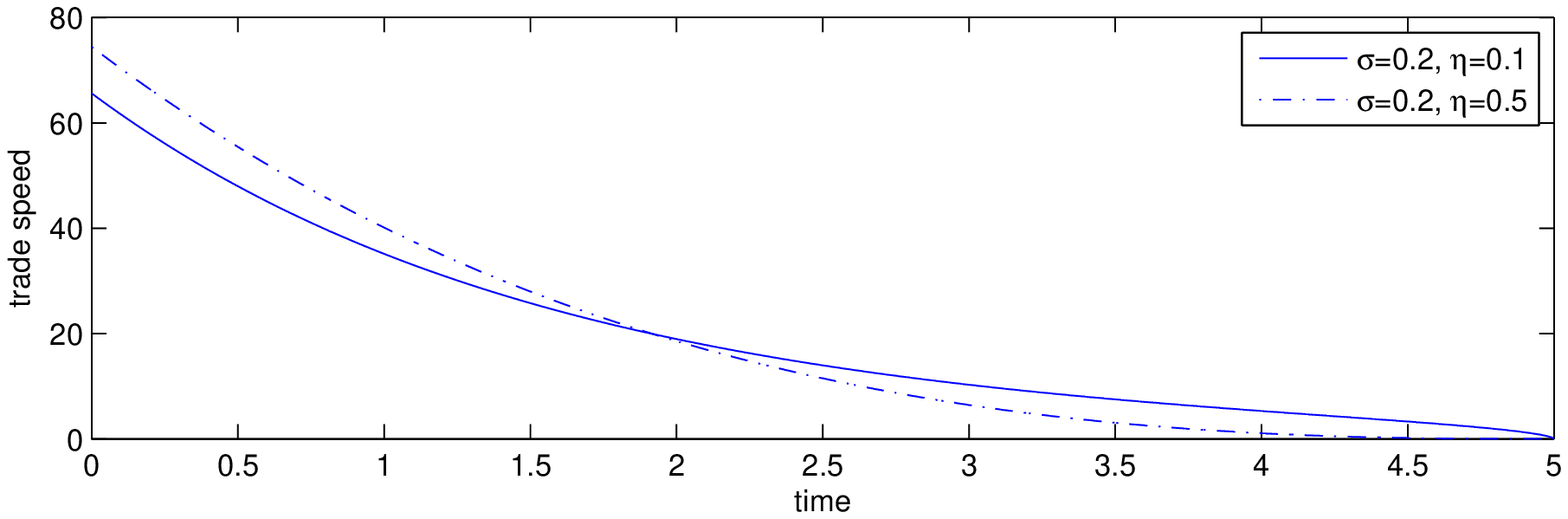}}}
\subfigure[]{
\resizebox*{13cm}{!}{\includegraphics{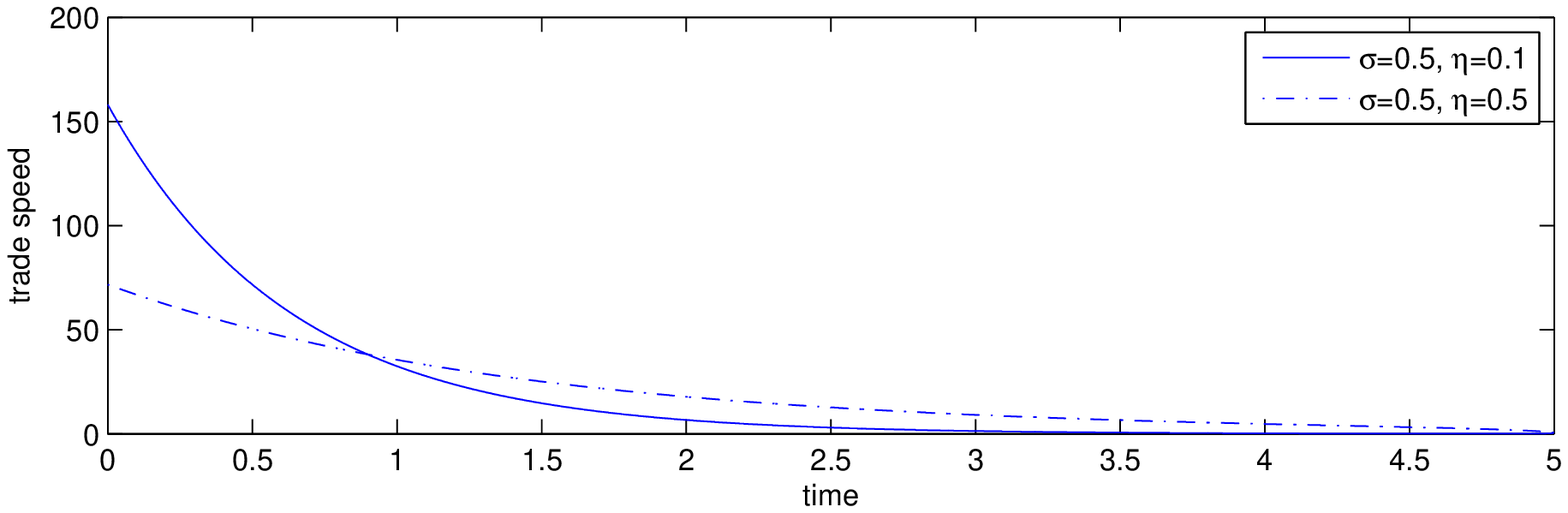}}}
\end{center}
\end{figure}

In Figure 2, we present the trading strategy on four simulated paths of
stock price and pricing signals. For simplicity we assume that the
time-varying parameters to be constant and we summarize the values for
various parameters as $S_0=\$100, \sigma_1=\sigma_2=0.5, \eta=0.1$. Figures
(2a), (2c), (2e) gives the trading speed in the case $\alpha_0=\$102,
\theta=0.2$, Figures (2b), (2d), (2f) gives that in the case $\alpha_0=\$98,
\theta=0.2$, Figures (3a), (3c), (3e) gives that in the case $%
\alpha_0=\$102, \theta=0.05$ while Figures (3b), (3d), (3f) gives that in
the case $\alpha_0=\$102, \theta=0.05$. When $\theta =0$, the model with
random signal degenerates to be the basic one.

\begin{figure}[]
\caption{Trade Strategy with Random Pricing Signals(I)}
\begin{center}
\subfigure[]{
\resizebox*{7.8cm}{!}{\includegraphics{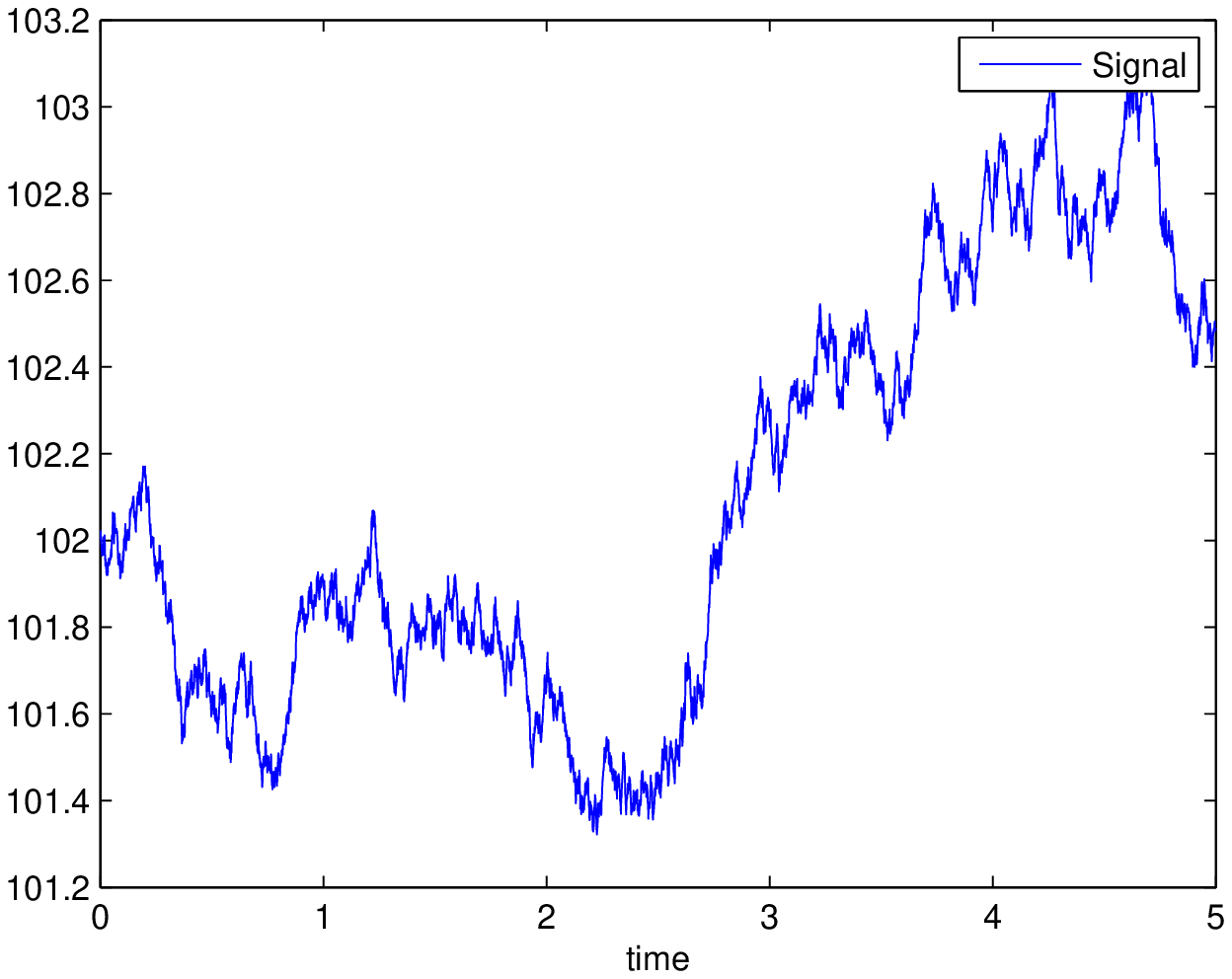}}}
\subfigure[]{
\resizebox*{7.8cm}{!}{\includegraphics{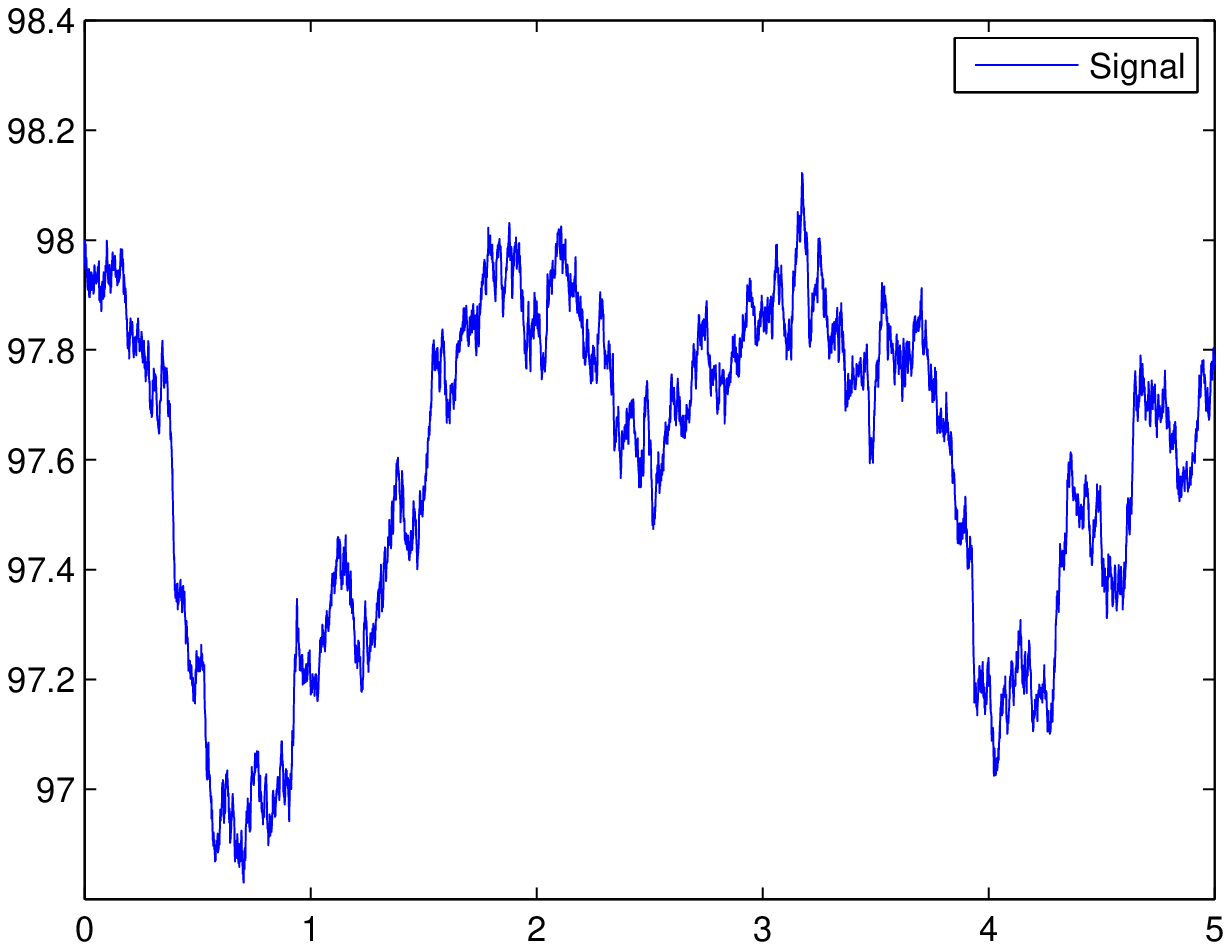}}}\\[0pt]
\subfigure[]{
\resizebox*{7.8cm}{!}{\includegraphics{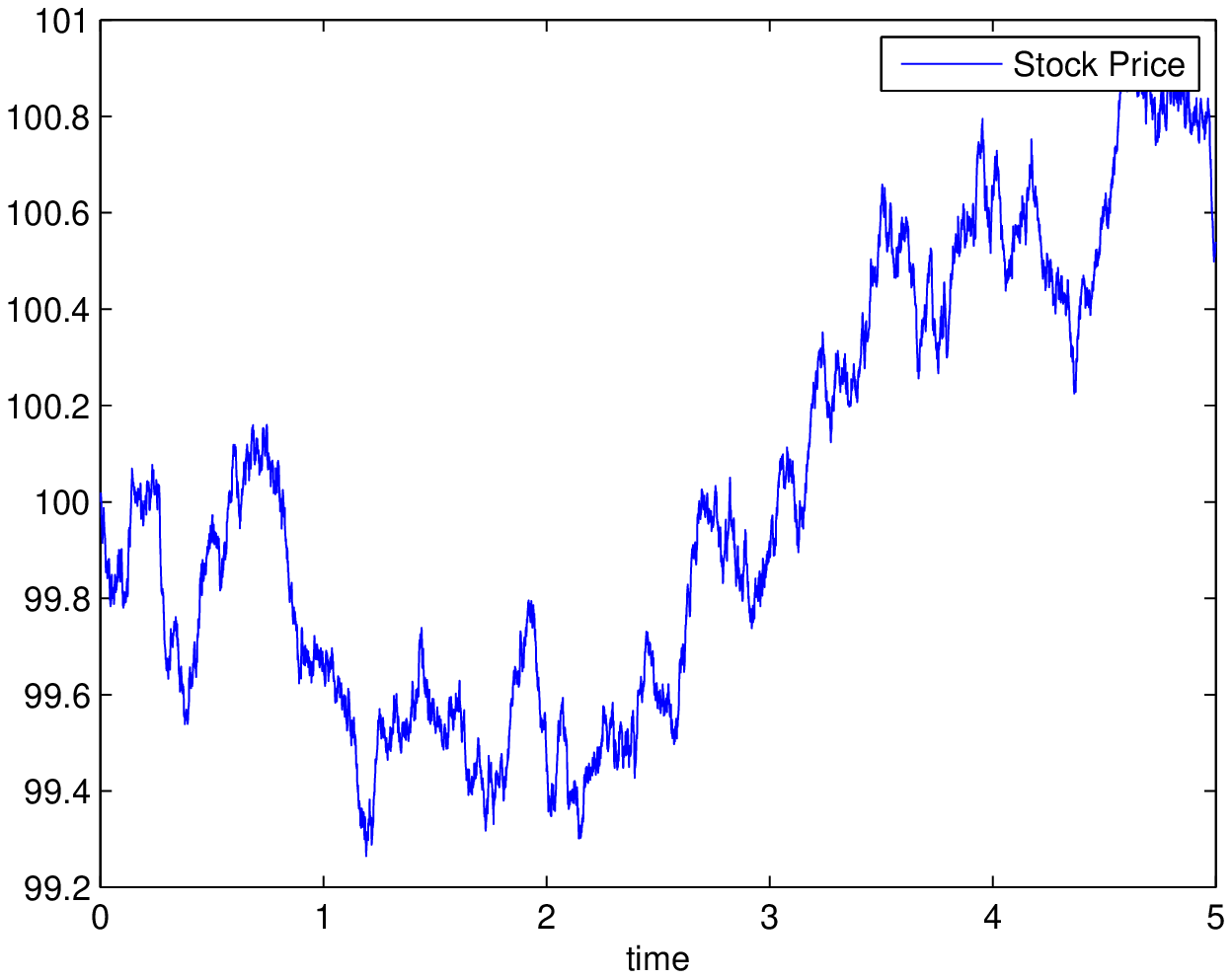}}}%
\subfigure[]{
\resizebox*{7.8cm}{!}{\includegraphics{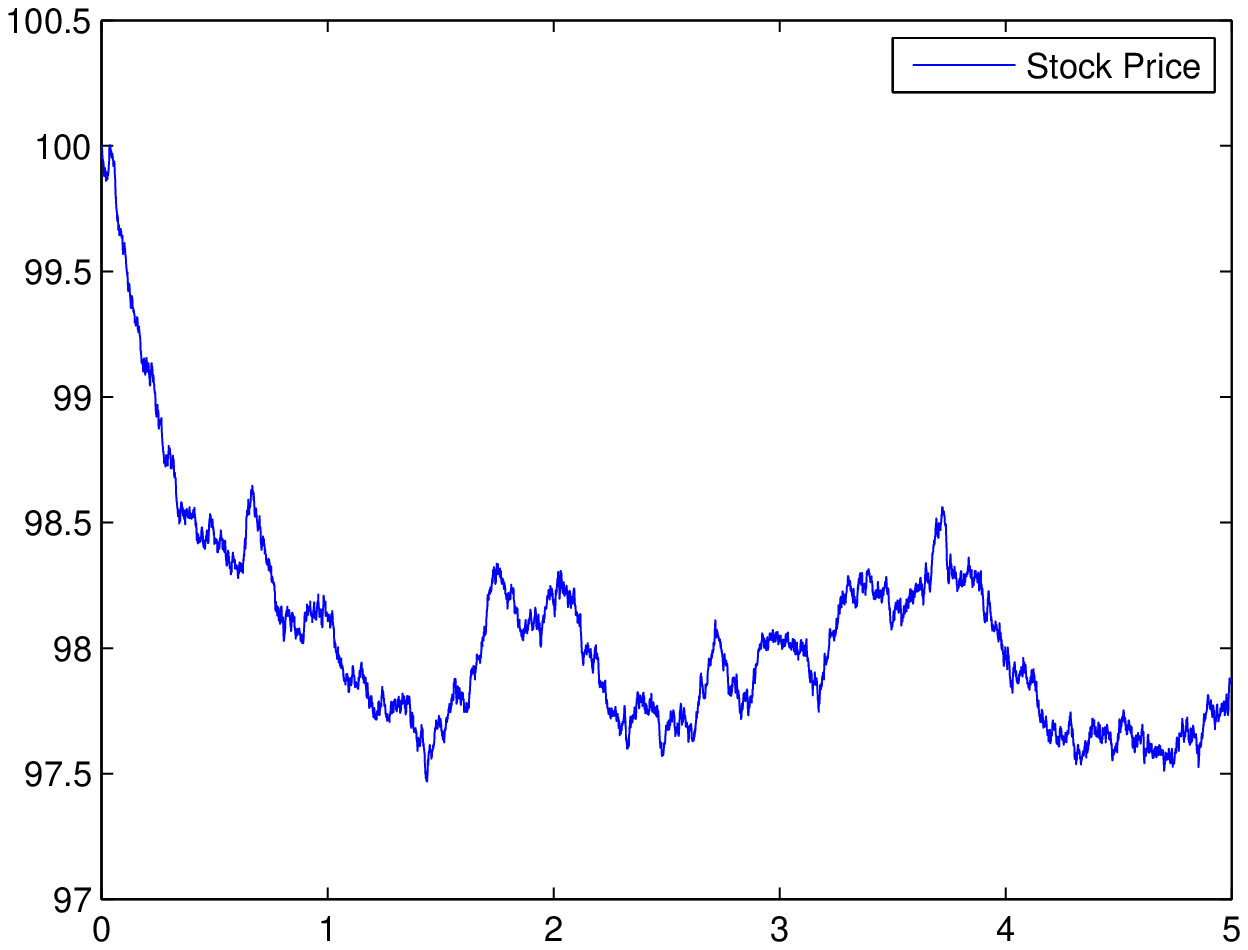}}}\\[0pt]
\subfigure[]{
\resizebox*{7.8cm}{!}{\includegraphics{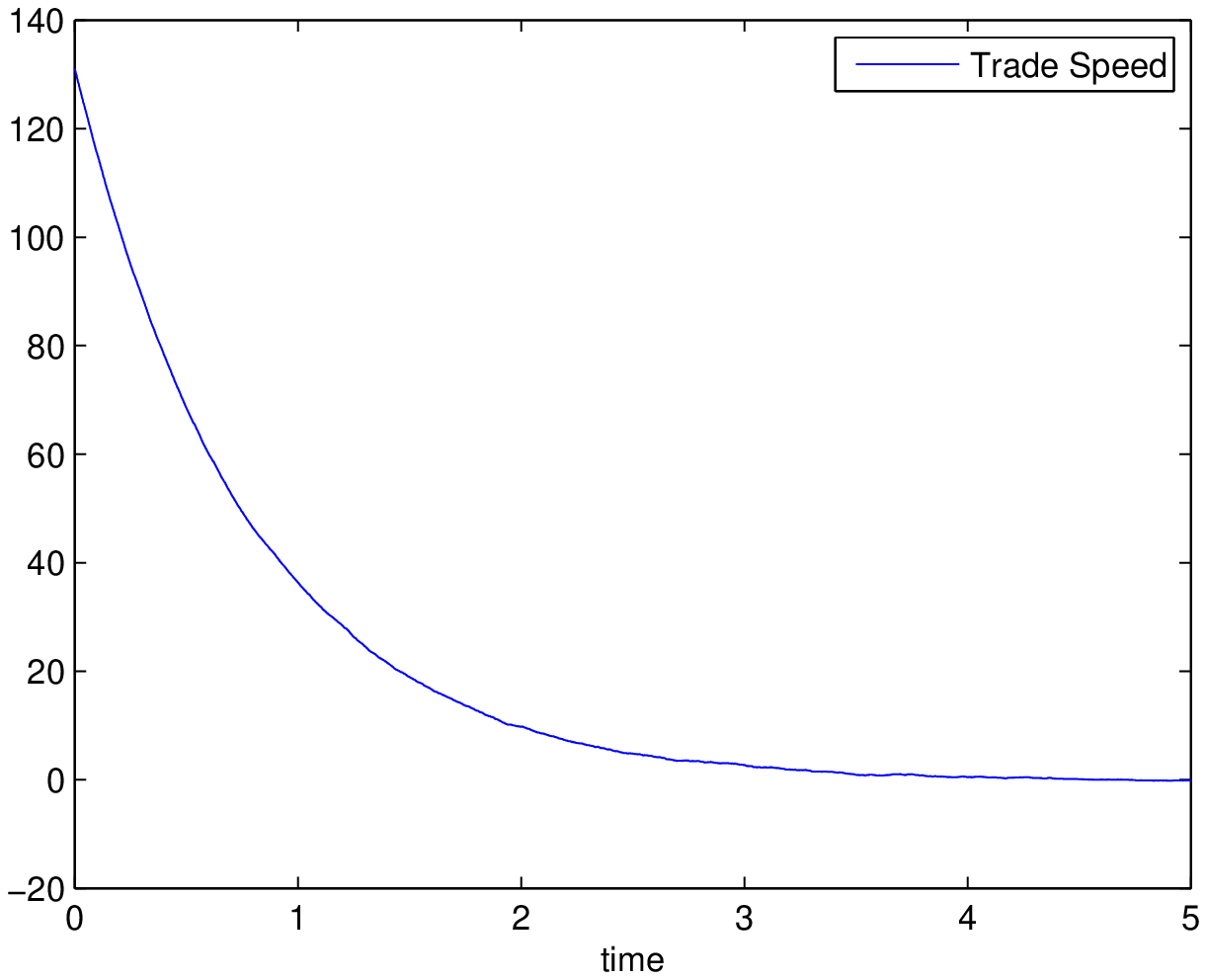}}}%
\subfigure[]{
\resizebox*{7.8cm}{!}{\includegraphics{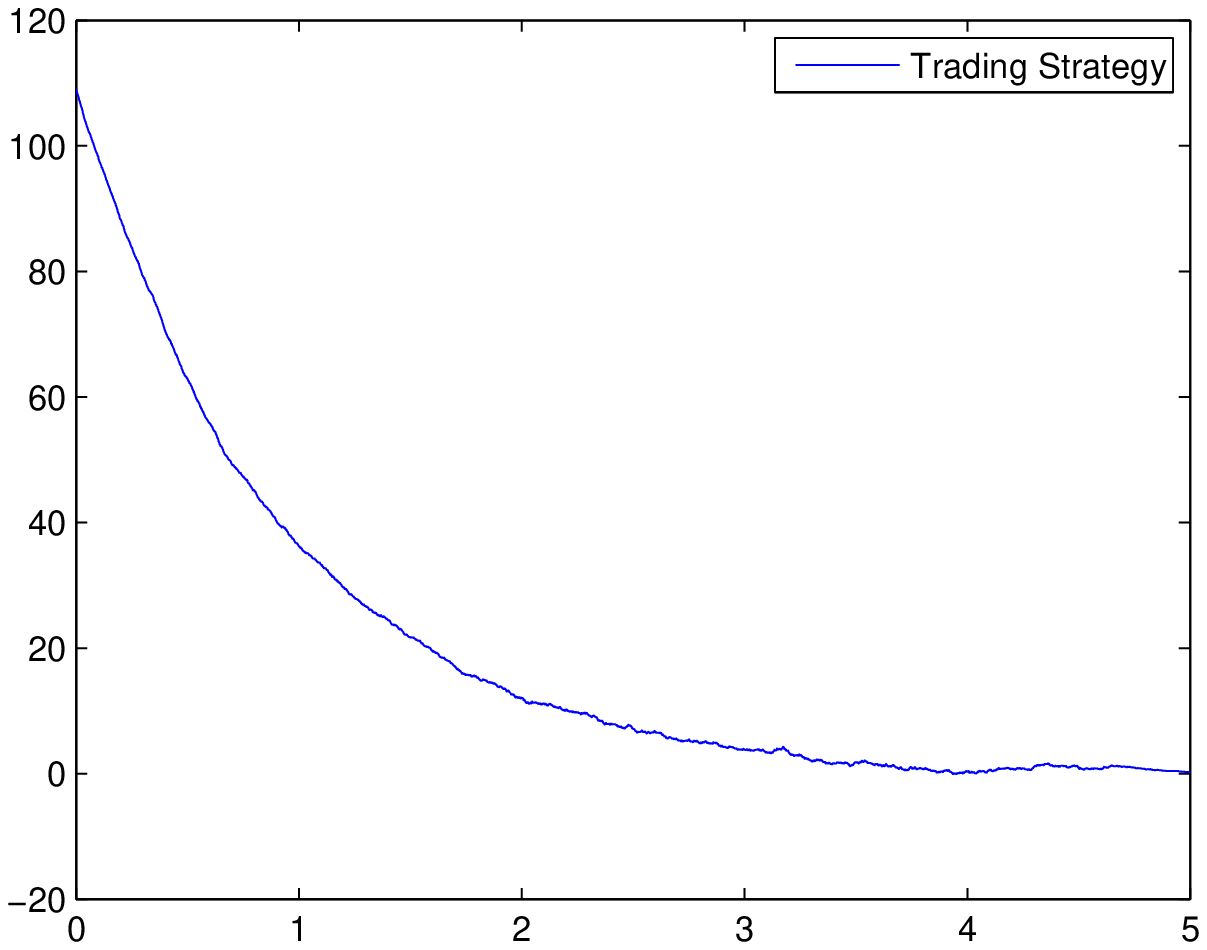}}}\\[0pt]
\end{center}
\end{figure}

\begin{figure}[]
\caption{Trade Strategy with Random Pricing Signals(II)}
\begin{center}
\subfigure[]{
\resizebox*{7.8cm}{!}{\includegraphics{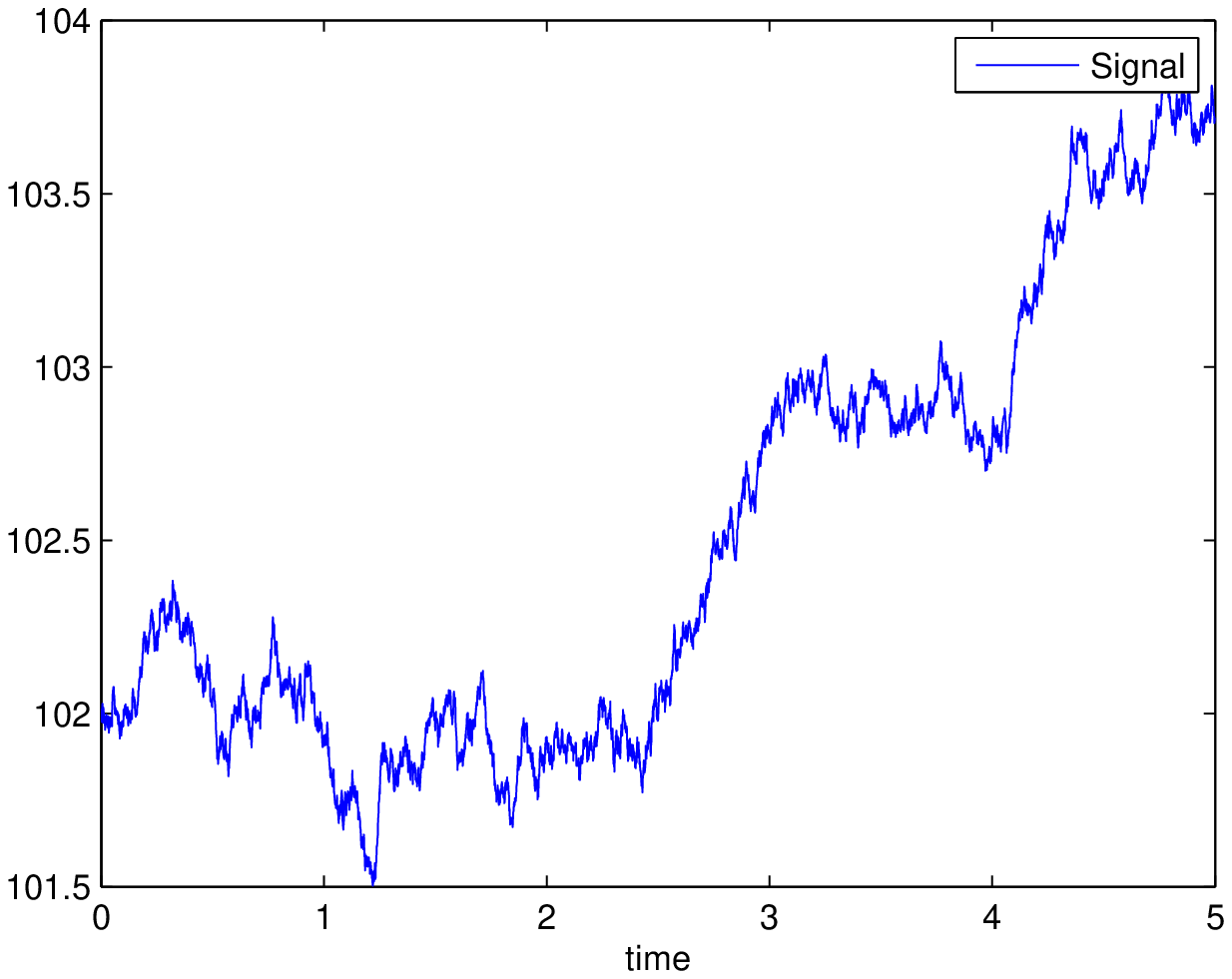}}}
\subfigure[]{
\resizebox*{7.8cm}{!}{\includegraphics{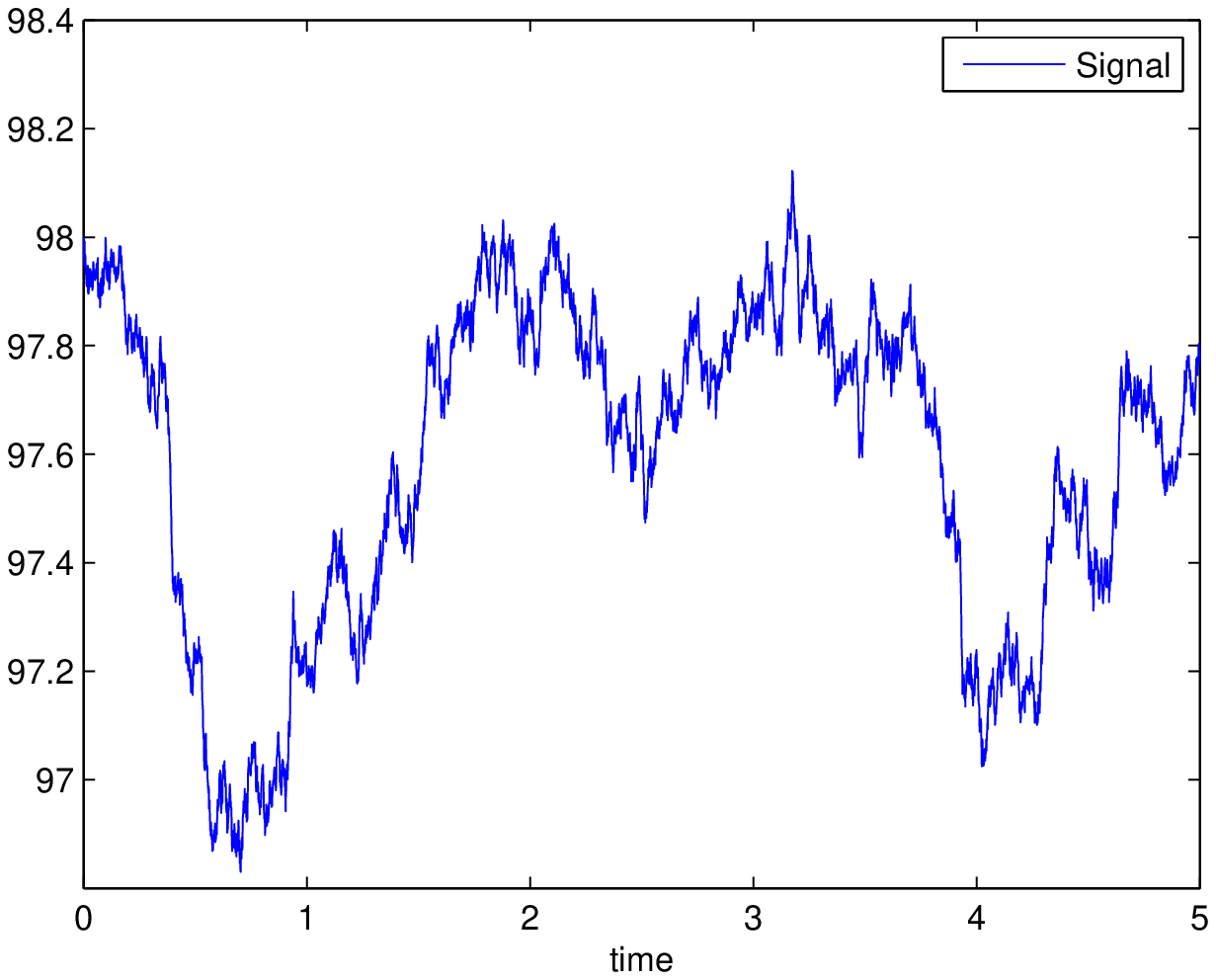}}}\\[0pt]
\subfigure[]{
\resizebox*{7.8cm}{!}{\includegraphics{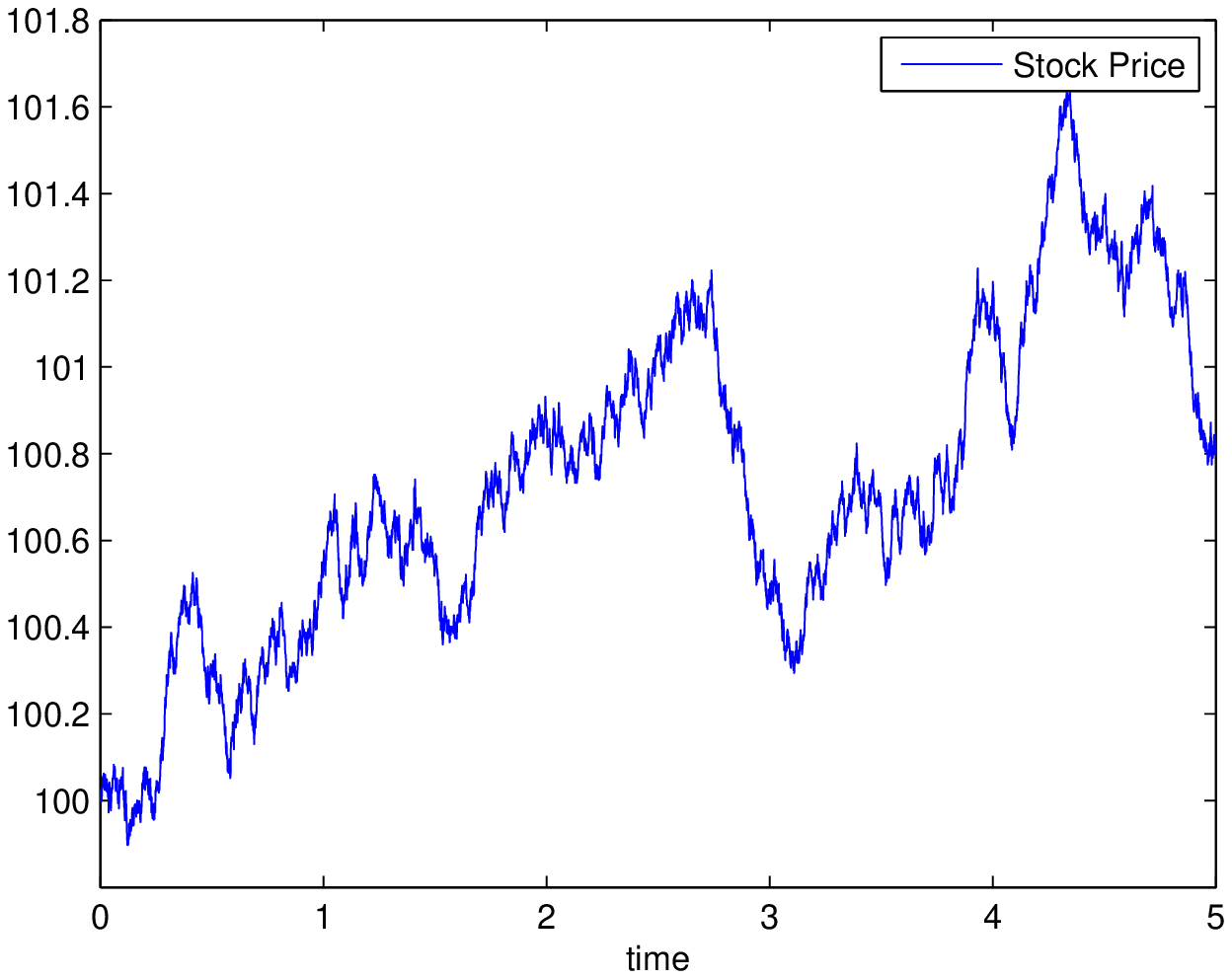}}}%
\subfigure[]{
\resizebox*{7.8cm}{!}{\includegraphics{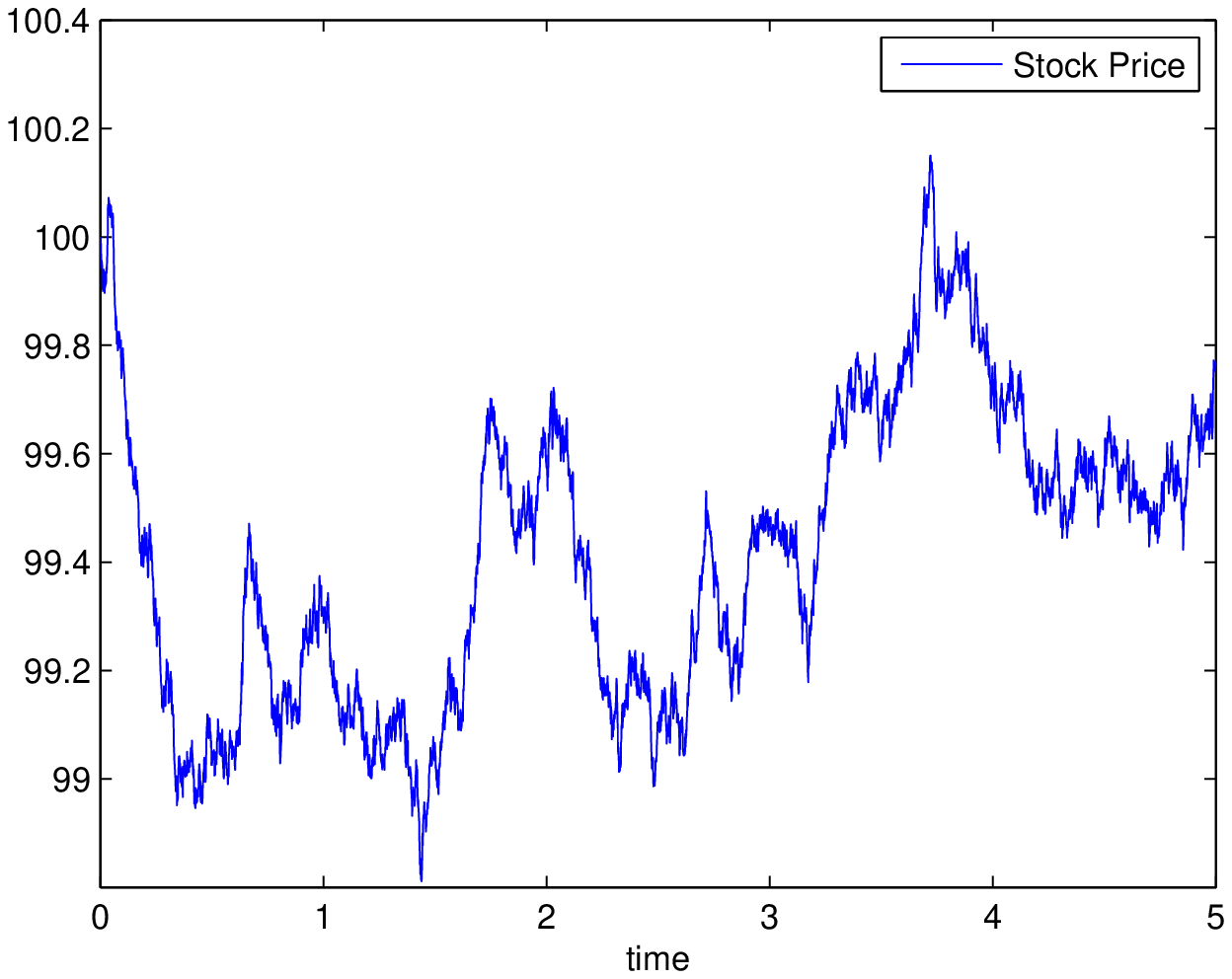}}}\\[0pt]
\subfigure[]{
\resizebox*{7.8cm}{!}{\includegraphics{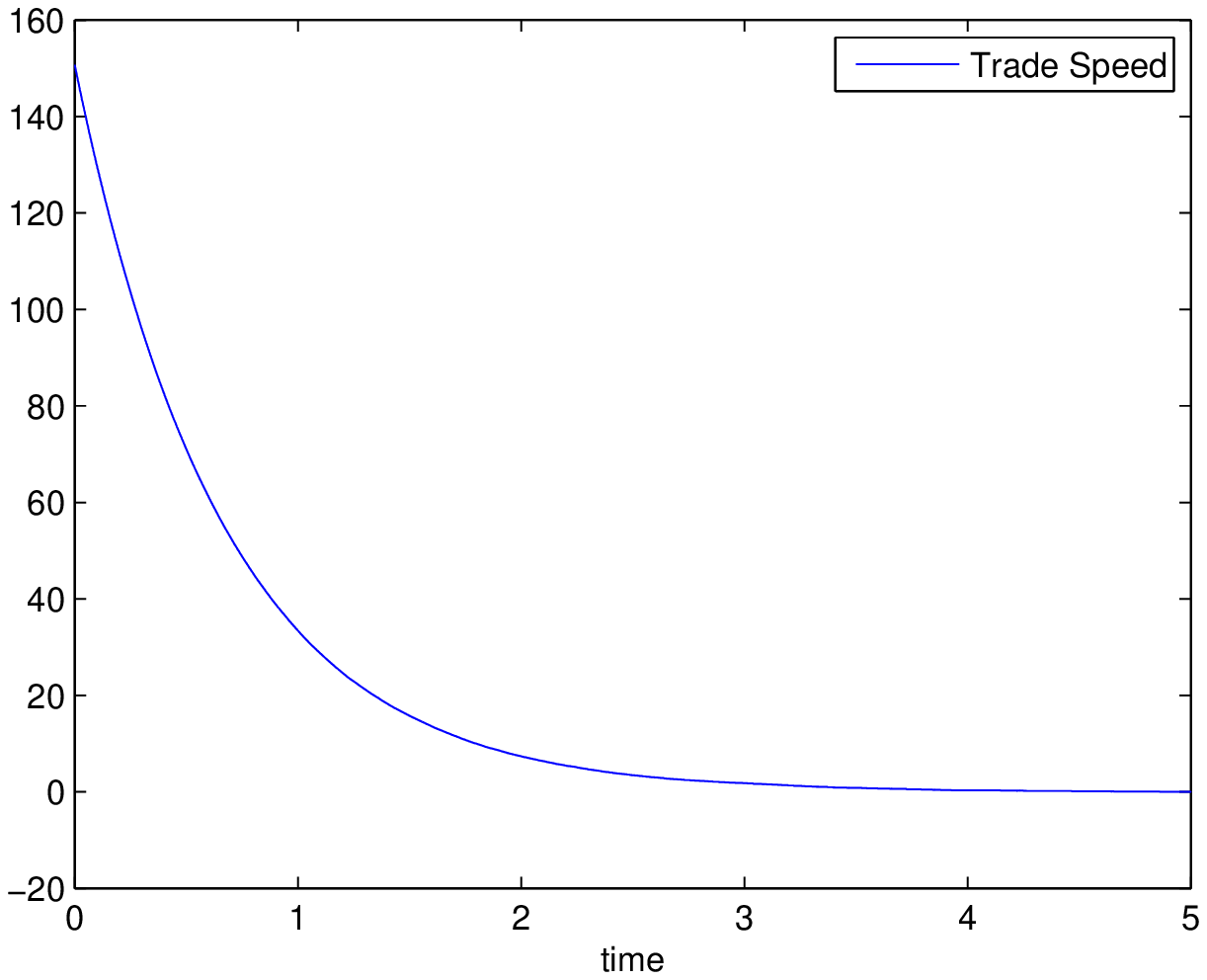}}}%
\subfigure[]{
\resizebox*{7.8cm}{!}{\includegraphics{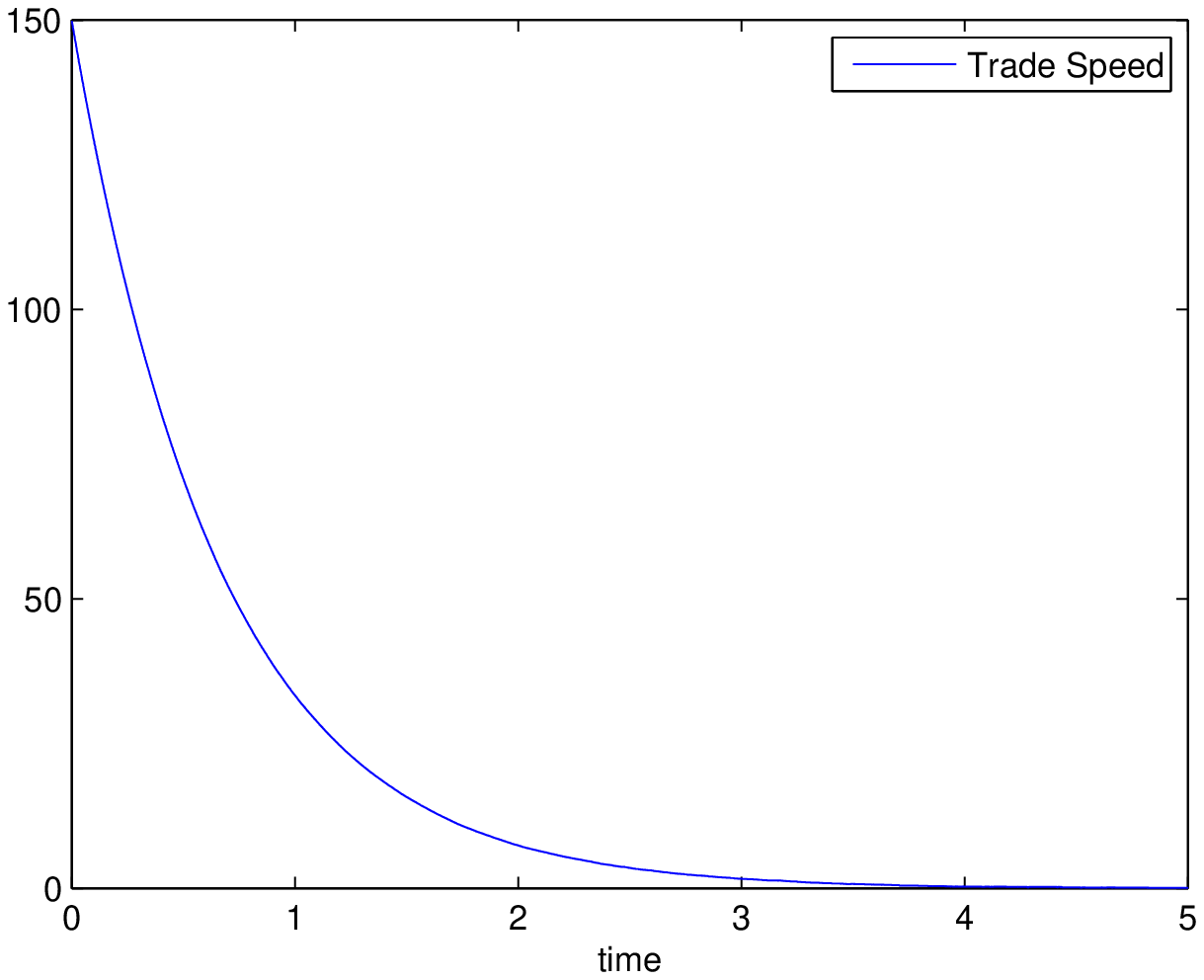}}}
\end{center}
\end{figure}

Figure 3 shows the trading strategy in the stochastic volatility model
(example 1). For simplicity, we assume $a_{\xi}$ and $b_{\xi}$ are constants
and we summarize the values of the parameters are as follows $a_{\xi}=0,
b_{\xi}=0.1, \eta=0.1, \xi_0=1, S_0=\$100$.

\begin{figure}[]
\caption{Trade Strategy in Stochastic Volatility Model}
\begin{center}
\subfigure[]{
\resizebox*{13cm}{!}{\includegraphics{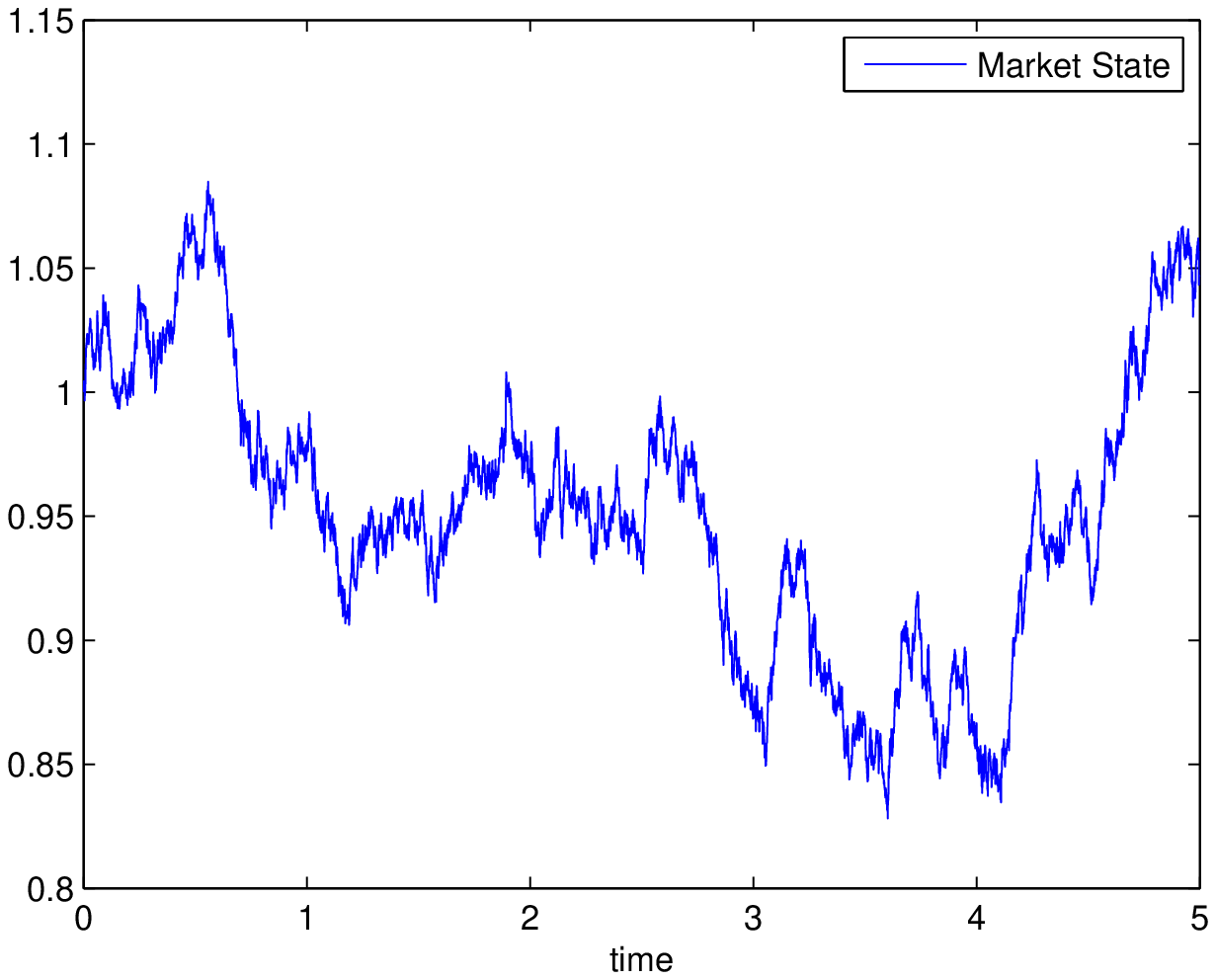}}}
\subfigure[]{
\resizebox*{13cm}{!}{\includegraphics{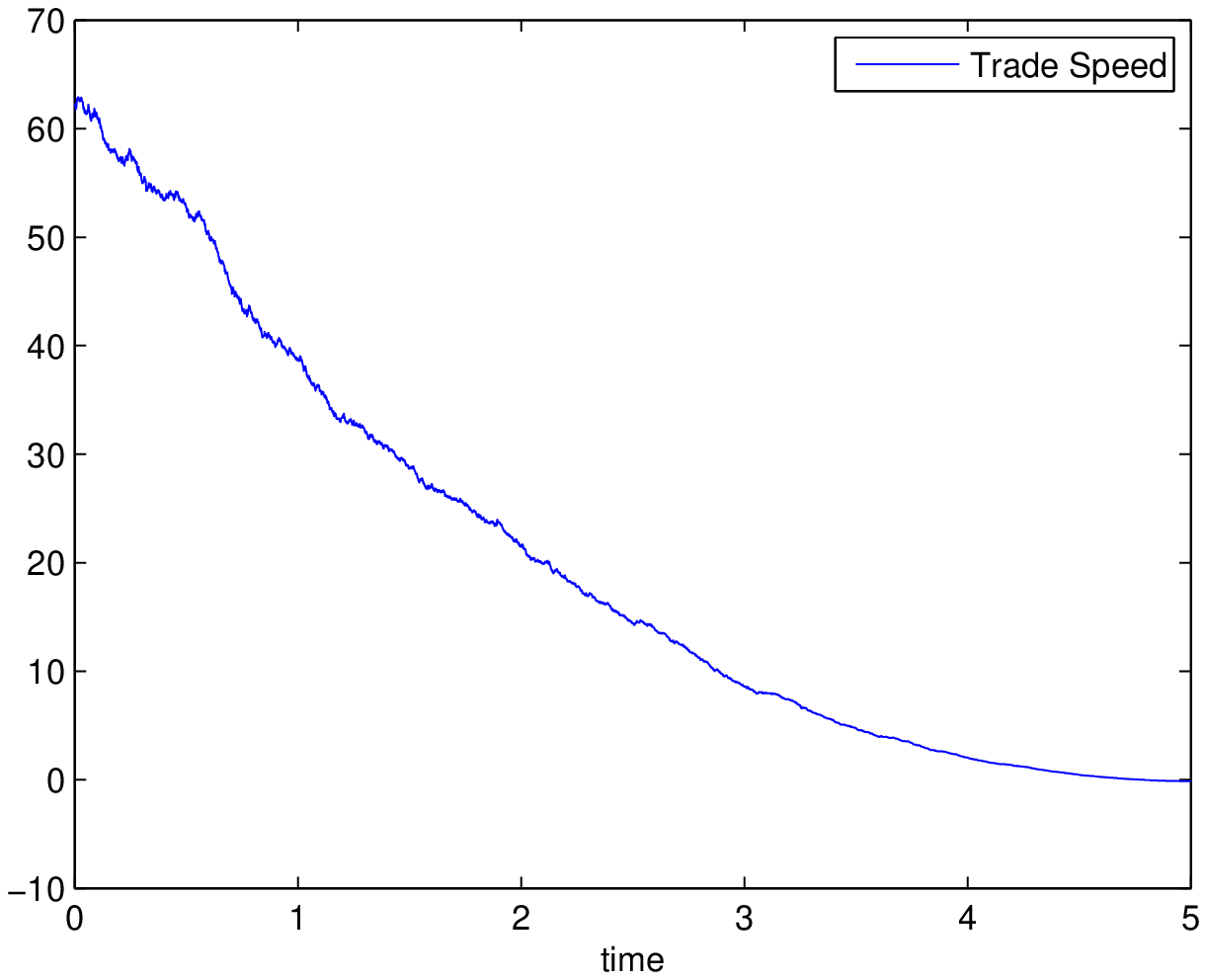}}}
\end{center}
\end{figure}

\section{Conclusions}

In this paper, we consider the quantitative trading problem under the
dynamic mean-variance criterion and derive time-consistent solutions in
three important models. We give a optimal strategy under a reconsidered mean-variance subject at any point in time.
 We also get an explicit trading
strategy when random pricing signals are incorporated. When consider
stochastic liquidity and volatility, we give the exact HJB equations. We
obtain an explicit solution in stochastic volatility model with a given
structure supported by empirical study.

\end{document}